\documentclass[11pt,a4paper,oneside, final]{amsart}

\usepackage{newpxtext}
\usepackage{newpxmath}
\renewcommand{\mathbb}[1]{\varmathbb{#1}}

\AtBeginDocument{%
  }
    
\newcommand{\hl}[1]{#1}

\usepackage{amsmath, amsthm, mathrsfs, calc, bm}
\usepackage{mathtools}
\usepackage{url}
\usepackage{mleftright}
\usepackage{enumitem}
\usepackage{microtype}
\usepackage{aliascnt}
\usepackage{hyperref}
\usepackage{pgf}
\usepackage{pgfplots}
\usepackage{tikz}
\usepackage{booktabs}
\usetikzlibrary{calc,decorations.markings,positioning,decorations.pathreplacing,calligraphy,shapes,arrows}
\usepgfplotslibrary{polar}\usepgflibrary{shapes.geometric}\usepgflibrary{plotmarks}

\usepackage[T1]{fontenc}
\usepackage[utf8]{inputenc} \DeclareUnicodeCharacter{0301}{í}

\usepackage{algorithm}
\usepackage{algorithmic}

\usepackage[capitalise,nameinlink]{cleveref}
\usepackage{caption}
\usepackage[labelfont=bf,textfont=normalfont,singlelinecheck=off,justification=raggedright]{subcaption}

\usepackage{threeparttable}
\usepackage{multirow}

\newcommand{\appref}[1]{\hyperref[#1]{Appendix~\ref*{#1}}}

\newtheorem{theorem}{Theorem}[section]

\newaliascnt{lem}{theorem}
\newtheorem{lemma}[lem]{Lemma}
\aliascntresetthe{lem}

\newaliascnt{cor}{theorem}
\newtheorem{corollary}[cor]{Corollary}
\aliascntresetthe{cor}

\newaliascnt{prop}{theorem}
\newtheorem{proposition}[prop]{Proposition}
\aliascntresetthe{prop}

\newaliascnt{claim}{theorem}

\aliascntresetthe{claim}

\newaliascnt{conjecture}{theorem}

\aliascntresetthe{conjecture}

\theoremstyle{definition}

\newaliascnt{example}{theorem}
\newtheorem{example}[example]{Example}
\aliascntresetthe{example}

\newaliascnt{defin}{theorem}

\aliascntresetthe{defin}

\newaliascnt{prob}{theorem}

\aliascntresetthe{prob}

\newaliascnt{rem}{theorem}
\newtheorem{remark}[rem]{Remark}
\aliascntresetthe{rem}

\newaliascnt{question}{theorem}

\aliascntresetthe{question}

\Crefname{question}{Question}{Questions}

\renewcommand{\epsilon}{\ensuremath\varepsilon}

\renewcommand{\Lambda}{\ensuremath\mu}

\newcommand{\Z}{\mathbb{Z}}
\newcommand{\N}{\mathbb{N}}

\newcommand{\Q}{\mathbb{Q}}
\newcommand{\K}{\mathbb{K}}

\newcommand{\NN}{\mathbb{N}}
\newcommand{\ZZ}{\mathbb{Z}}
\newcommand{\QQ}{\mathbb{Q}}

\newcommand{\KK}{\mathbb{K}}

\newcommand{\OO}{\mathcal{O}}

\newcommand{\diag}[1]{\ensuremath{\operatorname{diag}(#1)}}
\DeclareMathOperator{\GL}{GL}
\newcommand{\ideal}[1]{\left\langle #1 \right\rangle}

\providecommand*{\eu}%
{\ensuremath{\mathrm{e}}}
\providecommand*{\iu}%
{\ensuremath{\mathrm{i}}}

\newcommand{\seq}[3][{}]{\langle #2 \rangle_{#3}^{#1}}

\newcommand{\loopie}{\mathcal{L}}

\newcommand{\alg}{\overline{\mathbb{Q}}}

\newcommand{\Lexp}[1]{L_{\operatorname{exp}}(#1)}
\newcommand{\Col}[1]{#1 : \left(\prod_{i=1}^d x_i\right)^\infty} %
\DeclareMathOperator{\Sat}{Sat}
\DeclareMathOperator{\Row}{Row}
\DeclareMathOperator{\Null}{Null}

\usepackage[backend=biber,
			natbib=true, %
			style=numeric,
            url=false,
            doi=true,
            isbn=false,
            giveninits=true,
            maxbibnames=99]{biblatex}

\date{}

\usepackage[foot]{amsaddr}

\author{George Kenison$^1$}

\email{george.kenison@tuwien.ac.at}

\author{Laura Kov\'{a}cs$^1$}
\email{laura.kovacs@tuwien.ac.at}

\author{Anton Varonka$^1$}
\email{anton.varonka@tuwien.ac.at}

\address{$^1$TU Wien, Vienna, Austria}

\thanks{
Acknowledgements: We are grateful to Amaury Pouly and
James Worrell for valuable discussions. The work presented in this
paper was partially supported by the ERC consolidator grant ARTIST
101002685, the WWTF grant ProbInG ICT19-018,
and the EU Marie Sklodowska-Curie Doctoral Network LogiCS@TU Wien
Grant Nr.\ 101034440.
}

\begin{document}

\title[From Polynomial Invariants to Linear Loops]{From Polynomial Invariants to Linear Loops}

\begin{abstract}
Loop invariants are software properties that hold before and after every
iteration of a loop. As such, invariants provide 
inductive arguments that are key in automating the verification of
program loops.
The problem of
generating loop invariants; in particular, invariants
described by polynomial relations (so called \emph{polynomial
  invariants}), is therefore one of the hardest problems in software
verification.
In this paper we advocate an alternative solution to invariant
generation. Rather than inferring invariants from loops, we 
synthesise loops from invariants. As such, we generate loops that
satisfy a given set of polynomials; in other words, our synthesised
loops are correct by construction.

Our work turns the problem of loop synthesis into a symbolic
computation challenge. We employ techniques from algebraic
geometry to synthesise loops whose polynomial invariants are
described by \emph{pure difference} binomials. We show that such 
complex polynomial invariants need ``only'' linear loops, opening up
new venues in program optimisation. 
We prove the existence of non-trivial loops with linear updates for
polynomial invariants generated by pure difference
binomials. Importantly, we introduce an algorithmic approach that
constructs linear loops from such polynomial invariants, by 
generating linear recurrence sequences that have specified algebraic relations among their terms.
\end{abstract}

\maketitle

\section{Introduction} \label{sec:introduction}

\emph{Loop invariants}, or more simply \emph{invariants} in the
sequel,  are software properties that hold before and after every iteration of a loop.
Invariants are often key to inductive arguments for
automating the verification of programs with loops, see, 
e.g.~\cite{Rodriguez04,Kovacs08,Oliveira16,Kincaid18,Humenberger18}.
One challenging aspect in invariant synthesis is the derivation of
\emph{polynomial invariants} for loop programs over numeric data structures.
Such invariants 
are defined  by polynomial relations 
\(P(x_1, \dots, x_d) = 0\) among the program variables $x_1, \ldots,
x_d$. A nice property of polynomial invariants is that they define a
polynomial ideal, called a \emph{polynomial invariant
ideal}~\cite{Rodriguez04,Kovacs08}. As such, the problem of generating
(all) polynomial invariants is reduced to generating a finite basis of
the polynomial invariant ideal. 

\emph{In this paper, 
we reverse engineer the problem of polynomial invariant generation and
propose an alternative solution to invariant synthesis.} Rather than
generating polynomial invariants for a given loop, we synthesise loops
for a given polynomial invariant. %
\hl{Assume the postcondition of some section of code with a loop 
is given by a conjunction of polynomial equalities.
Instead of generating an invariant that implies the postcondition,
our solution comprises the synthesis of a new loop, 
one that is correct with respect to the specification--%
and that by construction.}

\noindent\textbf{Linear Loop Synthesis.} The key aspect of our work
comes with considering \emph{homogeneous linear loops}, hereafter
simply linear loops. 
Linear loops are a class of single-path loops whose update assignments are determined by a homogeneous system of linear equations in the program variables.
More specifically, a \emph{linear loop} is a loop program of the form
\begin{equation*}
	\mathcal{L} \colon \bm{x} \leftarrow \bm{s};\ \textrm{while}\ \star\ \textrm{do}\ \bm{x}\leftarrow M \bm{x},
\end{equation*}
where \(\bm{x}\) is a \(d\)-dimensional column vector of program variables, 
\(\bm{s}\) is a \(d\)-dimensional vector with rational entries, 
and \(M\) is a \(d\times d\)-matrix with rational entries.
Herein we employ the notation \(\star\), instead of using
\texttt{true} as loop guard, as our focus is on loop synthesis rather
than proving loop termination.

Linear loops are fundamental objects in the computational study of recurrence.
On the one hand, this class of loops represents a restrictive computational model.
On the other hand, fundamental problems such as the decidability of the Halting Problem are open for this class~\cite{ouaknine2015survey}.
For the avoidance of doubt, we lose no generality by working over the class of linear loops rather than the class of \emph{affine loops} (those single-path loop programs with update assignments of the form \(\bm{x}\leftarrow M \bm{x} + \bm{v}\) where \(\bm{v}\in\QQ^d\)).
Indeed,  the problem of studying
the functional behaviour of affine loops can be reduced  to that of studying linear
loops; admittedly the reduction step will, in general, increase the
number of program variables~\cite{Rodriguez04,Kovacs08,ouaknine2015survey}.\smallskip

\noindent\textbf{Linear Loops and Recurrences.} 
Given a (system of) polynomial relation(s), our work 
constructs a linear loop.
By construction, each of the given polynomial relations is satisfied by the loop variables before and after each
iteration; thus each relation is an invariant of the loop. 
(Specific details are given in the discussion on our contributions.)
We employ techniques from algebraic geometry to
synthesise loops whose polynomial invariants are described by
\emph{pure difference binomials}. 
For such polynomial invariants, the procedure in \cref{sec:binomial} shows that ``only” linear loops are required, thus addressing
challenging aspects of arithmetic reductions in program optimisation~\cite{SRed01}. 
Further, we prove the existence of non-trivial loops with linear updates for
polynomial invariants generated by pure difference binomials
(see \cref{sec:further}). 
\cref{runex1,runex2} showcase  
linear loops that are synthesised by our work. We note that the 
study of \emph{pure difference ideals}, those ideals generated by pure
difference binomials, and their respective linear loops is well motivated. %
The classes of lattice and toric ideals (defined in \cref{sec:preliminaries}) are pure difference ideals.
Pure difference ideals also appear in the encoding of random walks on Markov chains of the form \(\N^d\) by  way of the connected components of the underlying graph  \cite{diaconis1998lattice,kahle2014positive}.
\begin{example}\label{runex1}
	Consider the input ideal~$I_1 \subseteq \QQ[x,y,z]$ generated by the
        polynomials $p_1 = x^2-y$ and $p_2 = x^3-z$. Note that $p_1$
        and $p_2$ are pure difference binomials in $x,y,z$. 
	The procedure of~\cref{sec:binomial} outputs a linear loop~$\loopie_1$ (see \cref{subfig:loop1}).
\end{example}
	\begin{example}\label{runex2}
	Given the input ideal~$I_2 \subseteq \QQ[x,y]$ generated by the polynomial~$p = x^3y-xy^3$.
	The procedure of~\cref{sec:binomial} outputs a linear loop~$\loopie_2$ (see \cref{subfig:loop2}). %
Note, $I_1$ is the invariant ideal of~$\loopie_1$. Loop~$\loopie_2$, in turn, satisfies all invariants in~$I_2$; however, $I_2$ is a strict subset of the invariant ideal of $\loopie_2$.
\end{example}

\begin{figure}[t]
	\begin{subfigure}[b]{0.4\linewidth}
			\caption{Loop~$\loopie_1$.} \label{subfig:loop1}
		\begin{algorithmic}
			\STATE $(x, y, z):= (1, 1, 1);$
			\WHILE{$\star$}
			\STATE $x:=2x;$
			\STATE $y:=4y;$
			\STATE $z:=8z;$
			\ENDWHILE
		\end{algorithmic}%
	\end{subfigure}%
	\begin{subfigure}[b]{0.29\linewidth}
				\caption{Loop~$\loopie_2$.} \label{subfig:loop2}
		\begin{algorithmic}
			\STATE $(x, y):= (1, 1);$
			\WHILE{$\star$}
			\STATE $x:=2x;$
			\STATE $y:=2y;$
			\ENDWHILE
			\STATE
		\end{algorithmic}%
	\end{subfigure}
	\caption{Synthesised linear loops \(\loopie_1\) and \(\loopie_2\)}
	\label{loops}
	\end{figure}

Key to our work is modelling loops as linear recurrence sequences that have specified algebraic relations among their terms.
Let \(x_n\) denote the value of a loop variable
\(x\) at the \(n\)th loop iteration.
Recall that for a linear loop,
\(x_n\) is given by a linear recurrence sequence with constant
coefficients, commonly known as a \emph{C-finite} sequence~\cite{everest2003recurrence}.
C-finite sequences are therefore key to our loop synthesis procedure (\cref{sec:binomial}): we model program loops as systems of recurrence equations.
Indeed, we aim to synthesise a system of C-finite recurrence
sequences, and hence a linear loop, 
that satisfies given polynomial relations. \smallskip %

\noindent\textbf{Our Contributions.} \hfill%
\begin{enumerate}
	\item For a polynomial ideal, we consider the problem of synthesising a \emph{non-parametrised} loop (i.e., synthesising both the loop body and concrete initial values) such that every polynomial in the ideal is an invariant of the loop.
	In particular, we demonstrate a procedure for synthesising loops from pure difference ideals.  In fact,
		\begin{enumerate}
			\item Given a pure difference ideal \(I\), we describe a process that synthesises a linear loop with invariant ideal~\(I\) (\cref{sec:binomial}).
			\item Suppose that \(I=\langle p_1,\ldots, p_k\rangle\) is a polynomial ideal not necessarily generated by binomials, for which, by a change of coordinates, there exists a generating set of pure difference binomials.
			We present a procedure that outputs a linear loop such that any polynomial \(p\in\langle p_1,\ldots, p_k\rangle\) is an invariant of said loop (\cref{sec:to-bin}).
		\end{enumerate}
	\item Under reasonable assumptions about the input, the aforementioned generated loops are \emph{non-trivial}.
By non-trivial, we mean that the orbit (or trajectory) of the vector 
	\begin{equation*}
		\seq[\infty]{\bm{x}(n)}{n=0} = \langle x_1(n) , \ldots, x_d(n) \rangle_{n=0}^\infty
	\end{equation*}
given by the values of the loop variables \(x_1,\ldots, x_d\) at the \(n\)th loop iteration is infinite.
That is to say, the loop program functions on an infinite-state system.
We shall defer a formal definition of a non-trivial loop to \cref{sec:preliminaries}.
\end{enumerate}

In \cref{sec:further}, we consider corollaries to our loop synthesis procedure from \cref{sec:binomial}.
In particular, we consider specialisations for restrictive classes of input ideals.
One corollary, concerning \emph{canonical pure difference binomials}, is as follows.
\begin{corollary} \label{cor:canonical}
Suppose that $I = \langle p \rangle \subseteq \QQ[x_1, \dots, x_d]$ where \(p\) is an irreducible polynomial of the form \(p=x_1^{\alpha_1} \cdots x_d^{\alpha_d} - x_1^{\beta_1} \cdots x_d^{\beta_d}\).
Then the procedure in \cref{sec:binomial} synthesises a linear loop for which \(I\) is precisely the invariant ideal of the loop.
\end{corollary}

We end this note with a conclusion that discusses related work and proposes directions for future research (\cref{sec:conclusion}).

\section{Preliminaries} \label{sec:preliminaries}

\subsection{Abstract Algebra}

Let \(\K\) be a field and \(\K^d\) the vector space of \(d\)-tuples in \(\K\).
A \emph{monomial} in the indeterminates (or polynomial variables) \(x_1, \ldots, x_d\) is an expression \(x_1^{v_1} x_2^{v_2} \cdots x_d^{v_d}\) where each exponent \(v_j\) is a non-negative integer.
We %
employ the shorthand \(\bm{x}^{\bm{v}}\) to denote such monomials. A \emph{polynomial} is a finite linear combination of monomials. This can be extended to include the negative integer exponents, if needed. We shall explicitly refer to monomials $\bm{x}^{\bm{v}}$ with $\bm{v} \in \ZZ^d$ (or their linear combinations) as \emph{Laurent monomials} (Laurent polynomials, respectively), whenever negative powers are considered.

Let \(\K[x_1,\ldots, x_d]\) be a polynomial ring, for brevity \(\K[\bm{x}]\). For computability reasons, throughout \(\K\) will be the field of algebraic numbers, so ~$\K:=\alg$.
We employ the standard notation~$\KK^*$ to refer to the set~$\KK \setminus \{0\}$ of invertible elements of the field and denote by $\GL_d(\KK)$ the multiplicative group of $d\times d$ invertible matrices over~$\KK$.

\subsection{Algebraic Geometry}

We recall standard preliminary material and terminology from the field of algebraic geometry, and refer to~\cite{cox2011toric,cox2015ideals} for more details. 

\subsubsection*{Ideals}

A \emph{polynomial ideal} is a subset $I \subseteq \K[\bm{x}]$ that satisfies the following properties:
$0\in I$; \(I\) is closed under addition; and for each $p \in \K[\bm{x}]$ and $q \in I$, necessarily $pq \in I$.
For a set of polynomials $S \subseteq \K[\bm{x}]$, the \emph{ideal generated by $S$} is given by
\begin{equation*}
    I = \ideal{S} := \{ s_1 q_1 + \cdots + s_\ell q_\ell : s_j \in S, q_j \in \K[\bm{x}], \ell \in \N \}.
\end{equation*}

A polynomial ideal \(I\) is \emph{proper} if \(I\) is not equal to \(\K[\bm{x}]\), \(I\) is \emph{prime} if \(p\cdot q\in I\) implies that \(p\in I\) or \(q\in I\), and \(I\) is \emph{radical} if \(p^n \in I\) implies that \(p\in I\).

Key to our discussion will be the bases for polynomial ideals.
\begin{theorem}[Hilbert's Basis Theorem]
 Every ideal in \(\K[\bm{x}]\) has a finite basis.
\end{theorem}
Seminal work by Buchberger introduced Gr\"{o}bner bases for polynomial ideals, which permit the algorithmic computation of key properties of polynomial ideals~\cite{DBLP:journals/jsc/Buchberger06a,cox2015ideals}, including ideal  membership, ideal union/intersection, elimination ideals, and many more. 
A key property that we draw upon in our synthesis procedure is the computation of a basis for a \emph{saturation} of an ideal.
Given \(q\in\K[\bm{x}]\), we compute a basis for the \emph{saturation of \(I\) with respect to \(q\)}
	\begin{equation*}
		I:(q)^\infty := \{p\in\K[\bm{x}] : q^n p\in I \text{ for some } n\in\N\}.
	\end{equation*}
We also recall relevant structural properties for classes of polynomial ideals.
\begin{theorem}[{\cite[Theorem 6, Chapter 4.6]{cox2015ideals}}]\label{radical-mpd}
	Each radical polynomial ideal \(I\) in \(\K[\bm{x}]\) admits a unique decomposition as an intersection of finitely many prime ideals \(I = p_1 \cap \cdots \cap p_r\) such that for each pair \(p_i\nsubset p_j\). The decomposition is commonly referred to as the \emph{minimal decomposition of a radical ideal}. 
\end{theorem}
\subsubsection*{Varieties}

The notion of an algebraic variety generalises the concept of algebraic curves to \(d\) dimensions.
An (affine) \emph{algebraic} variety \(V\subseteq \K^d\) is the locus of points satisfying a system of polynomial equations.
In the sequel, we shall focus on the varieties associated with polynomial ideals.
Let \(I\) be a polynomial ideal in \(\K[\bm{x}]\).
The locus of points in \(\KK^d\) where the polynomials in \(I\) simultaneously vanish is called the \emph{variety of the ideal} and denoted by \(V(I)\).
Specifically,
	\begin{equation*}
		V(I) = \{\bm{x}\in\K^d : p(\bm{x})=0 \text{ for all } p\in I\}.
	\end{equation*}
A \emph{subvariety} is a subset of a variety that is, itself, a variety. A variety \(V\) is \emph{irreducible} if when written as a union of subvarieties \(V= V_1\cup V_2\) then necessarily either \(V_1=V\) or \(V_2=V\).
An affine variety is irreducible if and only if \(I(V)\) is a prime ideal.

Let \(V\subseteq \K^d\) be an affine variety. 
We endow \(V\) with the \emph{Zariski topology} by declaring that the closed sets of \(V\) are precisely the subvarieties of \(V\). 
In this way, we extend the definition of \emph{Zariski closure} as follows.
The \emph{Zariski closure} of a subset~$S \subseteq \K^d$ is defined to be the smallest affine variety containing \(S\).

\subsection{C-finite Sequences}
The class of \emph{C-finite} sequences~\cite{DBLP:series/tmsc/KauersP11} consists of the real-algebraic linear sequences \(\seq[\infty]{u_n}{n=0}\) that satisfy recurrence relations %
\begin{equation} \label{eq:rec}
	u_{n+d} = a_{d-1} u_{n+d-1} + \cdots a_0 u_n,
\end{equation}
with constant coefficients \(a_{d-1},\ldots, a_0\) such that \(a_0\neq 0\).
A C-finite sequence \(\seq{u_n}{n}\) that satisfies \eqref{eq:rec} is entirely determined by its initial values \(u_0, \ldots, u_{d-1}\).
The \emph{order} of \(\seq{u_n}{n}\) is the minimum length of the recurrence relations it satisfies.

	Given a system of C-finite sequences $\seq{x_1(n)}{n}$, \dots, $\seq{x_d(n)}{n}$, 
	consider the ideal consisting of all polynomials~$P \in \KK[y_1, \dots, y_d]$ such that the polynomial equality~$P = 0$ is satisfied for all~$n \geq 0$ after simultaneously setting $y_j := x_j(n)$ for all~$j \in \{1, \dots, d\}$.
	We call this ideal the \emph{ideal of algebraic relations} over~$\KK$ associated with the aforementioned system.

\subsection{Loop Invariants}\label{sec:prelim-invariants}

Let $\loopie$ be a linear loop with variables~$\bm{x} = (x_1, \dots, x_d)$.
For each loop variable \(x_j\), let \(\seq{x_j(n)}{n}\) denote the sequence whose
\(n\)th term is given by the value of \(x_j\) after the \(n\)th loop iteration. 
A \emph{polynomial invariant} of~$\loopie$ is a polynomial~$P \in \KK[\bm{x}]$ such that \[P(x_1(n), \dots, x_d(n)) = 0\] holds for all~$n\geq 0$. 
The set of polynomial invariants of \(\mathcal{L}\) forms an ideal, called the \emph{(polynomial) invariant ideal} of \(\mathcal{L}\)~\cite{Rodriguez04,Kovacs08}.
In other words, the invariant ideal of~$\loopie$ is the ideal of algebraic relations over~$\KK$ among the sequences~\(\seq{x_1(n)}{n},\ldots, \seq{x_d(n)}{n}\). 
We note that, in our setting, it is always possible to compute a finite basis for the invariant ideal using Gröbner bases computation~\cite{kauers2008algrel,Kovacs08}.

\hl{The invariant ideal~$I$ of~$\loopie$ is radical.} It can be equivalently described by its variety~$V(I) \subseteq \KK^d$. 
Let \[\OO:= \left\{ (x_1(n), \dots, x_d(n)) \in \KK^d : n \in \NN \right\}\] 
be the \emph{orbit of the loop}.
Informally speaking, $\OO$ is the set of variable vectors that~$\loopie$ can reach during its execution. 
The Zariski closure of~$\OO$ in $\KK^d$ is precisely the variety~$V(I)$
and, further, is the smallest algebraic variety that contains~$\OO$.
The Zariski closure of $\OO$ is the \emph{strongest algebraic invariant} of~$\loopie$~\cite{Ouaknine20}.
We shall sometimes abuse terminology and refer to the Zariski closure of \(\OO\) as the \emph{Zariski closure of \(\mathcal{L}\)}.
A loop~$\loopie$ is \emph{trivial} if its orbit~$\OO$ is a finite set.

\subsection{Lattices}
Let $(G, +)$ be an additive abelian group. A set $X \subseteq G$ is \emph{linearly independent} if for any $n_1, \dots, n_k \in \ZZ$ and any pairwise distinct $a_1, \dots, a_k \in X$ a linear combination $\sum_{i=1}^k n_ia_i = 0$ is zero only if $n_1 = \dots = n_k = 0$. If $\langle X \rangle = G$, the set $X$ is called a \emph{basis} of $G$.

An abelian group that has a finite basis is referred to as \emph{lattice}. 
It is known that a lattice has fixed basis size, called \emph{rank} and, furthermore, every lattice of rank~$r \geq 1$ is isomorphic to~$\ZZ^r$ as a group. 
A \emph{saturation} of a lattice~$L \subseteq \ZZ^d$ is a sublattice of~$\ZZ^d$ 
defined as \[\Sat(L) := \{\bm{u} \in \ZZ^d : c \bm{u} \in L \text{ for some } c \in \ZZ \setminus \{0\}\}.\] 
If~$L = \Sat(L)$, the lattice~$L$ is said to be \emph{saturated}.

An important class of lattices are those that describe the multiplicative relations among algebraic numbers. 
The \emph{exponent lattice} of $\zeta_1, \dots, \zeta_k \in \KK^*$ is given by \[\Lexp{\zeta_1, \dots, \zeta_k} := \left\{(n_1, \dots, n_k) \in \ZZ^k : \prod_{i=1}^k \zeta_i^{n_i} = 1\right\}.\]

The \emph{lattice ideal} $I_L \subset \KK[x_1, \dots, x_d]$ of a lattice~$L \subseteq \ZZ^d$ is the ideal
\(I_L:= \ideal{\bm{x}^{\bm{\alpha}} - \bm{x}^{\bm{\beta}}  : \bm{\alpha}, \bm{\beta} \in \NN^d, \bm{\alpha} - \bm{\beta} \in L }\).

\subsection{Toric Ideals and Varieties}

Consider an integer $s \times d$-matrix \begin{equation*}
	A:= \begin{pmatrix}
		a_{11} & a_{12} & \dots & a_{1d} \\
		a_{21} & a_{22} & \dots & a_{2d} \\
		& & \dots & \\
		a_{s1} & a_{s2} & \dots & a_{sd}
	\end{pmatrix}
\end{equation*}
with column set $\{\bm{a}_1, \dots, \bm{a}_d \} \subset \ZZ^s$. %
The purpose of defining this matrix is twofold. 

\emph{First}, we can view \(A\) as a linear transformation 
	$T_A \colon \ZZ^d \rightarrow \ZZ^s$ such that $T_A(\bm{x}) = A\bm{x}.$
Thus we define the \emph{kernel of \(A\)} by \(\ker{A}:=\bigl\{ \bm{x}\in\ZZ^d : T_A(\bm{x})=\bm{0} 	\bigr\}\). It is easy to see that $\ker{A} \subseteq \ZZ^d$ is a lattice. In particular, one can define a lattice ideal for~$\ker{A}$. Moreover, $\ker{A}$ is saturated.

\emph{Second}, we can consider the columns of~$A$ as Laurent monomials over the indeterminates~$z_1, \dots, z_s$. 
Formally, a column~$\bm{a}_i$ of~$A$ defines a Laurent monomial as
\(
	\chi^i(z_1, \dots, z_s) = z_1^{a_{1i}}  \dots  z_s^{a_{si}}
\).
We define the ring homomorphism~
	\begin{equation*}
		\pi_A : \KK[x_1, \dots, x_d] \rightarrow \KK[z_1^{\pm 1}, \cdots, z_s^{\pm 1}]
	\end{equation*}
by mapping the variables  $x_1, \dots, x_d$ to $\chi^1, \dots, \chi^d$. 
The codomain of the homomorphism is the ring of Laurent polynomials over~$\KK$ with~$s$ variables.

The \emph{toric ideal}~$I_A$ associated with matrix~$A$ is the kernel of~$\pi_A$. The ideal~$I_A$ is a pre-image of a prime ideal, and hence $I_A$ is itself a prime ideal of $\K[\bm{x}]$. 
As shown in~\cite[Proposition 1.1.9]{cox2011toric}, a toric ideal has an equivalent definition as a lattice ideal
\begin{equation}\label{toric-as-kernel}
	I_A = \langle \bm{x}^{\bm{u}} - \bm{x}^{\bm{v}} : \bm{u},\bm{v}\in\N^d,\, \bm{u} - \bm{v} \in\ker{A} \rangle;
\end{equation} 
that is, as an ideal of a lattice~$\ker{A}$. We emphasise that not every lattice is a kernel of a linear transformation and respectively, lattice ideals are not necessarily toric. In fact, toric ideals are precisely the prime lattice ideals~\cite[Proposition~1.1.11]{cox2011toric}.

The homomorphism of $\KK$-algebras is naturally associated with a group homomorphism~$\varphi_A \colon \KK^s \rightarrow \KK^d$, defined by
	\begin{equation*}
		\varphi_A(\bm{z}) = (\chi^1(\bm{z}), \dots, \chi^d(\bm{z})).
	\end{equation*}
The vanishing set of a toric ideal~$I_A$ is the Zariski closure of~$\varphi_A(\KK^s)$. Such irreducible varieties of~$\KK^d$ are called \emph{toric varieties}. They share a number of important properties and have been widely studied along with their ideals \cite{cox2011toric,sturmfels1996gbcp}.

We refer to binomials of the form $\bm{x}^{\bm{u}} - \bm{x}^{\bm{v}}$ as \emph{pure difference binomials}. A \emph{pure difference ideal} is generated by pure difference binomials. 
Note that lattice and, in particular, toric ideals are pure difference ideals. Let~$p = x_1^{\alpha_1} \cdots x_d^{\alpha_d} - x_1^{\beta_1} \cdots x_d^{\beta_d}$ be a pure difference binomial. Its \emph{exponent vector} is $(\alpha_1 - \beta_1, \dots, \alpha_d - \beta_d) \in \ZZ^d$. 

\begin{remark} \label{rem:canonical}
Clearly, there are infinitely many pairwise distinct pure difference binomials sharing the same exponent vector
---the vector is invariant when we multiply a pure difference binomial by a monomial, e.g., $x^2 - y^2$ and $x^3y - xy^3$. 
For a given exponent vector~$\bm{v} = (v_1, \dots, v_d) \in \ZZ^d$, 
we define a unique \emph{canonical (pure difference) binomial} 
as $p(x_1, \dots, x_d) := \bm{x}^{\bm{{v_+}}} - \bm{x}^{\bm{{v_-}}}$, where 
	\begin{equation*}
	\bm{{v_+}}=(\max\{v_1, 0\}, \dots, \max\{v_d, 0\}) \in \NN^d
	\end{equation*}
and $\bm{{v_-}} = \bm{{v_+}} - \bm{v} \in \NN^d$. 
Note that $x^2-y^2$ is the canonical binomial associated with the vector $(2, -2)$, 
whereas $x^3y-xy^3$ (see also~\cref{runex2}) is not canonical.
\end{remark}
\section{From pure difference ideals to linear loops} \label{sec:binomial}

We now describe our loop synthesis approach, by restricting our input
polynomials/ideals to pure difference binomials/ideals.
Given a pure difference ideal \(I\) as an input, our
procedure outputs  a linear loop such that each polynomial in \(I\) is
an invariant of the loop. We thus show that \emph{loop synthesis is
decidable for pure difference polynomials/ideals}.

We next outline our loop synthesis procedure and illustrate its main
steps in \cref{tblofex}.

\subsection*{Loop Synthesis Procedure:} %

\hrule
\vspace{1mm}
\begin{enumerate}[leftmargin=*, wide=0pt]
	\item[{\bf Input:}] A list $p_1, \dots, p_k \in \KK[x_1, \dots, x_d]$ of pure difference binomials or, alternatively, a polynomial ideal \(I\) generated by pure difference binomials.
	\item[{\bf (Step 1)}] Construct the lattice $L$ determined by the exponents associated with the generators of \(I\).
	\item[{\bf (Step 2a)}] Analyse lattice ideals $I_L$ and $I_{\Sat(L)}$; the ideal $I_{\Sat(L)}$ is prime, enabling Step 2b.
	\item[{\bf (Step 2b)}] Construct a matrix \(A\) with \(\ker A = \Sat(L)\).
	\item[{\bf (Step 3)}] Synthesise a linear loop \(\mathcal{L}\) with a diagonal update matrix~$M$ whose diagonal entries are determined by the columns of \(A\).  The ideal \(I_{\Sat(L)}\) is the invariant ideal of \(\mathcal{L}\).
	\item[{\bf Output:}]  Linear loop \(\mathcal{L}\) with invariants~$p_1, \dots, p_k$.
\end{enumerate}
\hrule
\begin{table}[h]
	\begin{threeparttable}
		\caption{Loop synthesis for our running examples}
		\label{tblofex}
		\begin{tabular}{|l||l|l|} \hline
			 & \textbf{\cref{runex1}} & \textbf{\cref{runex2}}\\
			\hline		
			Input: $I$ & $\ideal{x^2-y, x^3-z}$ & $\ideal{x^3y-xy^3}$\\[0.1cm]\hline
			Step 1: $I_L$ & $\ideal{x^2-y, x^3-z}$ & $\ideal{x^2-y^2}$\\[0.1cm]\hline
			Step 2a: $I_{\Sat(L)}$ & $\ideal{x^2-y, x^3-z}$ & $\ideal{x-y}$\\[0.1cm]\hline
			Step 2b: $A$ & $\begin{pmatrix*} 1 & 2 & 3 \end{pmatrix*}$ & $\begin{pmatrix*} 1 & 1 & 1 \end{pmatrix*}$ \\[0.15cm]\hline
			\multirow{4}{*} {Step 3: $\loopie$} & $(x,y,z):=(1,1,1);$ & $(x,y):=(1,1);$ \\
			 & \textbf{while}\ $\star$\ \textbf{do} & \textbf{while}\ $\star$\ \textbf{do} \\ 
			 &  \ $(x, y, z):= (2x,4y,8z);$ & \ $(x, y):= (2x,2y);$\\
			 & \textbf{end while} & \textbf{end while} \\
			\hline
		\end{tabular}
	\end{threeparttable}
\end{table}
We further detail each component of our synthesis procedure,
proving correctness of each step of the synthesis process.

\subsection*{Input} %
The input for the synthesis process is a (possibly empty) finite list $p_1,
\dots, p_k$ of pure difference binomials.
Let $I$ denote the \emph{pure difference} polynomial ideal $I := \ideal{p_1, \dots, p_k}$ in~$\KK[\bm{x}]$ generated by this finite set. 
Our goal is to synthesise a linear loop \(\mathcal{L}\) for which~$I$ is an invariant.

\subsection*{Step 1: Lattice ideal \texorpdfstring{$I_L$}{IL}} %
We start by listing the $k$ exponent vectors $\bm{b}_1, \dots, \bm{b}_k \in \ZZ^d$ of the pure difference binomials $p_1, \dots, p_k$. 
Let $B := \{\bm{b}_1, \dots, \bm{b}_k\}$ and \(L\) the lattice spanned by the exponent vectors \(B\).
If $k=0$, our convention is $B:= \{\bm{0}\}$ and $L = \{\bm{0}\} \subseteq \ZZ^d$. 
In order to meet our objective, we first show that the saturation of the ideal
~$I = \ideal{p_1, \dots, p_k}$ 
with respect to~$\prod_{i=1}^{d}x_i$ 
is precisely the lattice ideal~$I_L$ (\cref{ItoIL}).
In this direction, we begin with an intermediate lemma.
\begin{lemma}\label{IandJsat}
	Let~$I$ be an ideal generated by~$p_1, \dots, p_k$ with exponent vectors constituting the set~$B = \{\bm{b}_1, \dots, \bm{b}_k\}$. 
	Let $J$ be the ideal generated by the canonical binomials of~$B$ (as in \cref{rem:canonical}). 
	Then \[\textstyle \Col{I} = \Col{J}.\]
\end{lemma}
\begin{proof}
	Since~$I \subseteq J$, one inclusion is straightforward: namely,
	\[\textstyle \Col{I} \subseteq \Col{J}.\]
	In the other direction, consider a polynomial~$g \in \Col{J}$. 
	Let $q_1, \dots, q_k$ be the canonical binomials generating~$J$. 
	Then $g$ admits a decomposition of the form
	\[\textstyle g \cdot \left(\prod_{i=1}^d x_i \right)^{n_1} = f_1 \cdot q_1 + \dots + f_k \cdot q_k\]
	for some polynomials~$f_1, \dots, f_k \in \KK[\bm{x}]$ and some $n_1 \geq 0$. 
	
	Each generator~$p_\ell$ of the ideal~$I$ has the form $p_\ell = m_\ell q_\ell$, where $m_1, \dots, m_k$ are monomials and so
	\begin{align*} \textstyle
		\left(\prod_{i=1}^{d}m_i \right) \cdot g \cdot \left(\prod_{i=1}^d x_i \right)^{n_1} =& \textstyle \sum_{\ell=1}^k  m_\ell q_\ell \cdot \left(f_\ell \cdot\prod_{i \neq \ell}m_i \right) \\ =& \textstyle \sum_{\ell=1}^k p_\ell \cdot \left(f_\ell \cdot\prod_{i \neq \ell}m_i\right) 
	\end{align*}
	lies in \(I\).
 Finally, there exists~$n_2 \geq 0$ such that $\left(\prod_{i=1}^{k}x_i\right)^{n_2}$ is a multiple of the monomial $\prod_{i=1}^{k}m_i$.
 Thus
\[ \textstyle \left(\prod_{i=1}^k x_i \right)^{n_2} \cdot g \cdot \left(\prod_{i=1}^k x_i \right)^{n_1} \in I,\] and hence~$g \in \Col{I}$.
\end{proof}

\begin{theorem}\label{ItoIL}
	Let $I \subseteq \KK[x_1, \dots, x_d]$ be a pure difference ideal and $L \subseteq \ZZ^d$ the lattice spanned by the exponent vectors of its generators.
	Then $\Col{I} = I_L$ and the lattice ideal~$I_L$ is radical.
\end{theorem}
\begin{proof}
	Since $B = \{\bm{b}_1, \dots, \bm{b}_k\}$ spans the lattice~$L$, \cite[Corollary 3.22]{herzog2018gtm} implies that $\Col{J} =I_L$. The desired equality $\Col{I} = I_L$ now follows from~\cref{IandJsat}.
	The second assertion follows from~\cite[Theorem 3.23]{herzog2018gtm}.
\end{proof}
Note that $I$ is contained in $I_L$ and, in general,
this inclusion is strict. 
In particular, $I \subsetneq I_L$ in~\cref{runex2}, 
see also~\cref{tblofex}.

\subsection*{Step 2a: Saturated lattice ideal \texorpdfstring{$I_{\Sat(L)}$}{ISat(L)}} %
Before we proceed with Step~2 of our loop synthesis process, we briefly reflect on the lattice ideals that contain~$I$. 
Notice that in Step~1 we not only found one of them, $I_L$, but we also pointed out that its lattice~$L$ had already been computed. 
Indeed, the generating set~$\bm{b}_1, \dots, \bm{b}_k$ of~$L$ is obtained from the pure difference binomials directly. 
While the generating polynomials~$I_L$ can also be computed
(e.g.\ by employing Gröbner bases techniques from~\cite[Section 5]{kesh2009ideals}
 to compute the saturation of~$I$),
our computation proceeds with vectors and lattices rather than polynomials and ideals.

Since \(I_L\) is a radical ideal (\cref{ItoIL}), \(I_L\) admits a unique decomposition (see~\cref{radical-mpd}). 
We adjust a general result by Eisenbud and Sturmfels~\cite{eisenbud1996binomial}, 
which concerns decompositions of binomial ideals,
to our setting.%
\begin{lemma}[{\cite[Corollary 2.5]{eisenbud1996binomial}}]\label{ILtoISatL}
	Let $L \subseteq \ZZ^d$ be a sublattice and $\Sat(L) \subseteq \ZZ^d$ its saturation. 
	The minimal decomposition of ideal \(I_L\) includes the lattice ideal $I_{\Sat(L)}$ so that %
	$I_L = I_{\Sat(L)} \cap J_1 \cap \dots \cap J_\ell$.
	Here $\ell \geq 0$ and $I_{\Sat(L)}, J_1, \dots, J_\ell$ are prime ideals of~$\KK[\bm{x}]$. 
	In particular, $I_{\Sat(L)}$ is toric.
\end{lemma} %

As an aside,
work by Grigoriev et al.~\cite{grigoriev2019biomodels} presents an alternative approach to this step.
Therein those authors decompose a so-called \emph{binomial variety}~$V(I_L)$  into a finite union of irreducible varieties.

In the computational part of Step~2 that follows, we compute a \emph{prime lattice ideal}~$I_{\Sat(L)}$ that contains~$I$.
\subsection*{Step 2b: Matrix \texorpdfstring{$A$}{A} encoding \texorpdfstring{$I_{\Sat(L)}$}{ISat(L)}} %
As before, $\{\bm{b}_1, \dots, \bm{b}_k\}$ is a set of vectors spanning~$L$ over~$\ZZ$.
Versions of the next proposition appear in the literature, see e.g.~\cite[Theorem 3.17]{herzog2018gtm}.
\begin{proposition}\label{constructA}
Let~$L \subseteq \ZZ^d$ be a sublattice of rank~$r$.
	If~$r < d$, then there exists an integer $(d-r) \times d$-matrix~$A$ such that $\Sat(L) = \ker{A}$. Otherwise, if~$r = d$, the same holds for $A = \bm{0} \in \ZZ^{1\times d}$.
	Further, there is an effective process to construct the matrix \(A\). 
\end{proposition}

	Recall that for each $\bm{y} := (y_1, \dots, y_d) \in\Sat(L)$ there is an integer~$c$ 
	such that $c\cdot \bm{y} \in L$. 
	Since~$L$ consists of integer linear combinations of $\{\bm{b}_1, \dots, \bm{b}_k\}$, 
	a vector of the form $c\cdot \bm{y}$ 
	lies in the $\ZZ$-span of the set~$\{\bm{b}_1, \dots, \bm{b}_k\}$.	
	Let \(V\) denote a vector subspace of~\(\Q^d\) over the field \(\Q\) spanned by $\{\bm{b}_1, \dots, \bm{b}_k\}$.
	Thus elements of \(\Sat(L)\) are the integer vectors in~\(V\).

	\autoref{constructA} follows from the next technical lemma.
	\begin{lemma} \label{claim:nullA}
		There is a computable integer matrix $A$ such that the null space, $\Null(A)=\{v \in \QQ^d : Av = 0\}$, is equal to~$V$.\footnote{Herein we distinguish between a $\QQ$-subspace of vectors orthogonal to the row space of a given integer matrix (commonly its \emph{null space}) and a sublattice of  orthogonal integer vectors (commonly its \emph{kernel}).}
		\end{lemma}
	\begin{proof}
	Let $B \in \ZZ^{k \times d}$ be the matrix with rows $\bm{b}_1, \dots, \bm{b}_k$. 
	The row space \(\Row(B)\) over \(\Q\) is thus equal to \(V\).
	The orthogonal complement of~$\Row(B)$ is $\Row(B)^\bot = \Null(B)$. 
	It follows that~$\Null(B)$ has a basis of $d-r$ linearly independent vectors, which we denote by~$\bm{a}_1, \dots, \bm{a}_{d-r}$, provided that $r < d$.
	On the other hand, if $r=d$, then \(V\) is not a proper subspace of \(\QQ^d\), from which it follows that $\Null(B) = \{\bm{0}\}$.
	Let \(A\) be the matrix with rows~$\bm{a}_1, \dots, \bm{a}_{d-r}$ (and defined as a single row vector $\bm{0}$ in the case $r=d$).	
	Then, by definition, $\Row(A) = \Null(B)$.
	The result follows by elementary properties of orthogonal decomposition in finite-dimensional vector spaces.
	Indeed, we have
	\begin{equation*}
		\Null(A) = \Row(A)^\bot = \Null(B)^\bot = (\Row(B)^\bot)^\bot = \Row(B) = V.
	\end{equation*}%
	
Generally speaking, the matrix $A$ constructed above has rational entries; however, multiplying the entries of a matrix by a scalar does not change the null space.
Thus we can assume, without loss of generality, that $A$ has integer entries and preserves the property $\Null(A) = V$, as desired.
\end{proof}
	\begin{proof}[Proof of \autoref{constructA}]
	There are two statements to prove in \autoref{constructA} depending on the cases \(r<d\) or \(r=d\).
	Note that the latter statement, that \(A=\bm{0}\in\Z^{1\times d}\) if \(r=d\), was dealt with in the proof of \cref{claim:nullA}.
	Thus all that remains it to establish the former statement.
	To this end, we show that \(\Sat(L)=\ker A\) for the matrix \(A\) constructed in the proof of \cref{claim:nullA}.
	We note that \(\Sat(L)\) is equal to
\begin{multline*}
		V \cap \ZZ^d = \Null(A) \cap \ZZ^d = \{v \in \QQ^d : Av = 0\} \cap \ZZ^d  \hfill \\ \hfill = \{v \in \ZZ^d : Av = 0 \} = \ker{A},
	\end{multline*}
	as desired.
	\end{proof}

\subsection*{Step 3: Synthesise linear loop \texorpdfstring{$\mathcal{L}$}{L}} %

We take a brief pause and reflect on the combination of \cref{ILtoISatL,constructA} together with the formal definitions presented in \cref{sec:preliminaries}.
The sum total is a threefold equivalence between i) toric ideals, ii) ideals of saturated lattices, and iii) ideals of lattices presented as~$\ker{A}$.
Our objective, to synthesise a loop with prescribed polynomial invariants, relies on the constructive aspect of each of the preceding steps.
In this final step,
we will use the matrix~\(A\), constructed in \cref{constructA},
to synthesise a linear loop with invariant ideal~\(I_{\Sat(L)}\).

In the proof of \autoref{toric-synt}, we will employ the following result due to Kauers and Zimmermann \cite[Proposition 5]{kauers2008algrel} concerning algebraic relations among C-finite sequences.
\begin{proposition} \label{prop:algrel}
The ideal over \(\overline{\Q}[x_0,x_1,\ldots, x_d]\) associated with the algebraic relations among the \(d+1\) bi-infinite C-finite sequences 
	\begin{equation*}
		\seq{n}{n\in\Z},\, \seq{\lambda_1^n}{n\in\Z},\,\ldots,\, \seq{\lambda_d^n}{n\in\Z}
	\end{equation*}
is equal to the lattice ideal of the exponent lattice of \(\lambda_1,\ldots, \lambda_d\).
\end{proposition}

We now turn to our main result in this step.
\begin{theorem}\label{toric-synt}
	Let $p_1, \dots, p_k \in \KK[x_1, \dots, x_d]$ be pure difference binomials and let~$L \subset \ZZ^d$ be the lattice spanned by their exponent vectors. There exists a linear loop~$\loopie$ with a diagonal update matrix $M\in\QQ^{d\times d}$ such that $I_{\Sat(L)}$ is the invariant ideal of~$\loopie$.
\end{theorem}
\begin{proof}
	Following \cref{constructA}, it remains to prove that the ideal of a saturated lattice~$\Sat(L)$ can be realised as the ideal of all polynomial invariants of a linear loop. 
	Our proof adapts the argument in \cite[Proposition 3.7]{galuppi2021toric} due to Galuppi and Stanojkovski.

	First, there exists an integer matrix~$A \in \ZZ^{s \times d}$ such that~$\ker{A} = \Sat(L)$ (by~\cref{constructA}). 
	As before (cf.~\cref{sec:preliminaries}), let $\{\bm{a}_1, \dots, \bm{a}_d\}$ be the set of column vectors of \(A\). 
	Now let $p_1, \dots, p_s$ be the first $s$ prime numbers and
	define~$\lambda_j:= p_1^{a_{1j}}\cdots p_s^{a_{sj}}$.
	Note that for each $j \in \{1, \dots, d\}$,  $\lambda_j$ is the evaluation of the monomial $\bm{x}^{\bm{a}_j}$ at~$\bm{p} = (p_1, \ldots, p_s)$. 
	
	Since there are no non-trivial multiplicative relations among pairwise distinct primes, 
	a vector $\bm{n}:= (n_1, \dots, n_d)$ is a member of the exponent lattice of~$\{\lambda_1, \dots, \lambda_d\}$ if and only if \(\bm{n}\) simultaneously achieves unity; i.e., $p_i^{n_1a_{i1} + \dots + n_da_{id}}=1$ for each \(i\in\{1,\ldots, s\}\). 
	Specifically, $(n_1, \dots, n_d) \in \Lexp{\lambda_1, \dots, \lambda_d}$ if and only if for each~$i \in \{1, \dots, s\}$ we have that $n_1a_{i1} + \dots + n_da_{id} = 0$. 
	Thus it follows, by definition, that $\Lexp{\lambda_1, \dots, \lambda_d} = \ker{A} = \Sat(L)$.
	
	Finally, let $J$ be the ideal generated by the algebraic relations among the C-finite sequences $\seq{\lambda_1^n}{n \in \NN}$, \dots, $\seq{\lambda_d^n}{n \in \NN}$. 
	\emph{Mutatis Mutandis,} the argument in \autoref{prop:algrel} holds for natural-indexed sequences.
	Taken in combination with the above argument, it follows that
	$J = I_{\Sat(L)}$.
	
	The polynomial invariants of a linear loop~$\loopie$ with initial vector $(1, \dots, 1) \in \QQ^d$ and update matrix~$M = \diag{\lambda_1, \dots, \lambda_d}$ are exactly those in the ideal~$J$, which concludes the proof.
\end{proof}

\subsection*{Output}

From our starting point of a pure difference ideal~$I$, we constructed an integer matrix~$A$ such that $I_{\Sat(L)}$ is a toric ideal associated with~$A$, and $I \subseteq I_{\Sat(L)}$.

We output the linear loop \(\mathcal{L}\)
with the invariant ideal~\(I_{\Sat(L)}\).
Since $I \subseteq I_{\Sat(L)}$,
the loop \(\mathcal{L}\)
meets our objective: each of the polynomials in \(I\) is invariant under the action of \(\mathcal{L}\).
Correctness of the procedure
 follows from \cref{constructA,toric-synt}.

 The following corollary summarises our synthesis procedure.
\begin{corollary}\label{worrell}
	Let $p_1, \dots, p_k \in \KK[x_1, \dots, x_d]$ be a set of pure difference binomials and $I = \ideal{p_1, \dots, p_k}$. 
	There exists a linear loop~$\loopie$ with a diagonal update matrix $M\in\QQ^{d\times d}$ such that any polynomial $p \in I = \ideal{p_1, \dots, p_k}$ is an invariant of~$\loopie$. 
	There is a procedure to effectively construct loop  \(\mathcal{L}\).
\end{corollary}
\section{Transformed pure difference ideals} \label{sec:to-bin}
The procedure in~\cref{sec:binomial}  synthesises linear loops from ideals generated by pure difference binomials. 
In this section, we demonstrate that, subject to certain assumptions
on the input, we can also synthesise linear loops for non-binomial ideals.%

We call a polynomial ideal \(\ideal{p_1,\ldots, p_k}\in\K[\bm{x}]\) a \emph{transformed pure difference ideal} if there exists an invertible linear change of coordinates for which each \(p_j\) in the generating set is a pure difference binomial after the coordinate change.
Here, by an \emph{invertible linear change of coordinates} we mean that the associated change-of-basis matrix \(S\) is an element of \(\GL_{d}(\KK)\).
\begin{theorem}\label{move-to-bin}
	Let \(I\) be a transformed pure difference ideal with associated change-of-basis matrix \(S\in\GL_{d}(\K)\).
	Then there exists a computable linear loop \(\loopie\) such that \(I\) is an
        invariant of \(\loopie\).
\end{theorem}

	We take \(I=\ideal{p_1,\ldots, p_k}\) and \(S\) as above. 
	Further, let  $J := \ideal{q_1, \dots, q_k}$ be the associated pure difference ideal after the change of coordinates \(S\).
	The synthesis procedure summarised in \cref{worrell} outputs a loop \(\loopie'\) for which the pure difference ideal~$J = \ideal{q_1, \dots, q_k}$ is invariant. 
	More specifically, $\loopie'$ is the loop with update assignments described by a diagonal update matrix~$M$ and initial vector~$\bm{v}_{\text{in}} = (1, \dots, 1)$
	such that each polynomial in ideal \(J\) is an invariant of
        \(\loopie'\).

	The proof of \autoref{move-to-bin} is as an immediate corollary of the next lemma, which outputs the desired linear loop \(\mathcal{L}\).
	\begin{lemma} \label{claim:transformed}
	Let $\loopie$ be the linear loop with update assignment matrix~$S^{-1}MS$ and initial vector~$S^{-1}\bm{v}_{\text{in}}$. 
	Then every polynomial in \(I\) is an invariant of~$\loopie$.
\end{lemma}
		\begin{proof}%
	Let $\bm{e}_1$, \dots, $\bm{e}_d$ be the standard basis of~$\KK^d$. 
	After the change of coordinates \(S\),
	the new basis vectors are $S^{-1}\bm{e}_1$, \dots, $S^{-1}\bm{e}_d$ and
	 the coordinates of~$\bm{a} = x_1\bm{e}_1 + x_d\bm{e}_d$ are $S\bm{a} = (x'_1, \dots, x'_d)$. 
		 Clearly, \[
	\bm{a} \in \bigcap_{i=1}^{k} V(p_i) 
	\Leftrightarrow
	S\bm{a} \in \bigcap_{i=1}^{k} V(q_i).
	\]
	Hence, $\bm{v}_\text{in} \in \bigcap_{i=1}^{k} V(q_i)$ implies that
	$S^{-1}\bm{v}_\text{in} \in \bigcap_{i=1}^{k} V(p_i)$. 
	Moreover, the matrix $S^{-1}MS$ encodes the linear update
        assignments of~$\loopie'$ in the standard basis; thus  $\left(S^{-1}MS\right)^n S^{-1}\bm{v}_\text{in} \in \bigcap_{i=1}^{k} V(p_i)$ for all $n\geq 0$.
	Hence, each polynomial in \(I\) is an invariant of \(\loopie\).
	\end{proof}
\begin{example}\label{ex-not-bin}
	Let~$I = \ideal{p_1, p_2}$ be the ideal of~$\KK[x,y,z]$ generated by $p_1 = 4y^2+y-x$ and $p_2= 8y^3-x+z$. The change of coordinates described by the invertible matrix
	\(S:= \biggl(\begin{smallmatrix}
		0 & 2 & 0\\
		1 & -1 & 0\\
		1& 0 & -1 
	\end{smallmatrix}\biggr)\)
	yields $q_1 = (x')^2-y'$, $q_2 = (x')^3 - z'$. 
	Thus \(I\) is a transformed pure difference ideal and, after a change of coordinates, we have arrived at the ideal in \cref{runex1}.
	In the new coordinate system, we synthesise a loop~$\loopie'$ (\cref{worrell,subfig:loop1}) for which each polynomial in the pure difference ideal \(\ideal{q_1,q_2}\) is invariant. 
	In the old coordinate system, we proceed as in \cref{move-to-bin}.
	Our procedure synthesises the loop~$\loopie$ with update matrix~$S^{-1} \cdot \diag{2,4,8}\cdot S$ and initial vector~$S^{-1}\cdot(1,1,1)^T$ as follows:%
	\begin{algorithmic}
		\STATE $(x, y, z):= (3/2, 1/2, 1/2);$
		\WHILE{$\star$}
		\STATE $x:=4x-2y;$
		\STATE $y:=2y;$
		\STATE $z:=-4x-2y+8z;$
		\ENDWHILE
	\end{algorithmic}
	By construction, both polynomials~$p_1$ and $p_2$ are invariants of~$\loopie$, and thus so is any $p \in I = \ideal{p_1, p_2}$.
\end{example}

\begin{remark}
	The approach of~\cref{move-to-bin} relies on the
        existence of an appropriate change of coordinates that moves
        ideal~$I$ to a pure difference ideal~$J$ and further, that a 
        change-of-basis matrix~$S$ is given explicitly.
        The problem of determining the existence of such~\(S\) appears to be a challenging task. 

A similar challenge is addressed by Katthän et al.~\cite{katthaen2019binomial}. 
Their algorithmic approach~\cite[Algorithm 4.5]{katthaen2019binomial} can generate all matrices~$S \in \GL_d(\KK)$ that move an ideal~$I$ to a \emph{unital} ideal~$J$.
Recall that unital ideals are those generated by pure difference binomials and monomials, and hence pure difference ideals are a strict subclass. 
A straightforward approach to transform an ideal to a pure difference ideal is to apply the procedure of~\cite{katthaen2019binomial} and to manually check whether any of the transformed ideals output has a basis without monomials. 
If~$J$ is such an ideal, then \cref{move-to-bin} can be employed to
synthesise a linear loop for the initial ideal~$I$.%
\end{remark}
\section{Further results} \label{sec:further}

In this section, we gather together corollaries of our synthesis procedure from \cref{sec:binomial}.

\subsection{Existence of Non-Trivial Linear Loops} 
\cref{infloopexist,inftoricideal} next show that, subject to certain
restrictions on the input ideal, there are always non-trivial loops
that witness a given input as an invariant.%
\begin{theorem}\label{infloopexist}
	Let $I$ be an ideal generated by at most~$d-1$ pure difference binomials in~$\KK[x_1, \dots, x_d]$. There exists a non-trivial linear loop~$\loopie$ such that any polynomial~$p \in I$ is an invariant of~$\loopie$.
\end{theorem}
\begin{proof}
	In light of \cref{worrell}, it suffices to show that~\(\loopie\) generated by the procedure described in~\cref{sec:binomial} is non-trivial.
	
	Recall that a linear loop \(\loopie\) is trivial if the input vector has finite orbit.
	For such loops, the variable vector~$(x_1, \dots, x_d)$ is attained at two different iterations of~$\loopie$.
	Since, by construction, the update matrix \(M:= \diag{\lambda_1,\ldots, \lambda_d}\), the loop \(\loopie\) is trivial if and only if $\lambda_1^k = \dots = \lambda_d^k = 1$ for some~$k \geq 0$. 
	This implies that all diagonal entries of the loop update matrix~$M$ are roots of unity. However, by construction, the entries $\lambda_1, \dots, \lambda_d$ are positive rational numbers (see proof of \autoref{toric-synt}). 
	Therefore, $\lambda_1 = \dots = \lambda_d = 1$ is a necessary condition for~$\loopie$ to be trivial.
	
	From the proof of \cref{constructA}, the matrix~$A$ has rank~$0$ if and only if the associated lattice $L$ is spanned by $r=d$ linearly independent vectors. 
	By assumption, $L$ has rank at most~$d-1$ because it is generated by at most~$d-1$ exponent vectors. Thus $M \neq \diag{1, \dots, 1}$ and so \(\loopie\) is non-trivial.
\end{proof}
The next theorem is a specialisation of~\cref{infloopexist} and
follows from~Step~3 of the synthesis  procedure in~\cref{sec:binomial}.
\begin{theorem} \label{inftoricideal}
	Let $I_T \neq \ideal{x_1 - 1, \dots, x_d-1}$ be a toric ideal of $\KK[x_1, \dots, x_d]$. There exists a non-trivial linear loop~$\loopie$ such that $I_T$ is the invariant ideal of~$\loopie$.
\end{theorem}
\begin{proof}
	We recall that toric ideals are precisely the lattice ideals of saturated lattices~$\Sat(L) = \ker{A}$.
	Due to~\cref{toric-synt}, a loop~$\loopie$ with invariant ideal~$I_{\Sat(L)}$ can be generated for an arbitrary saturated lattice~$\Sat(L) = \ker{A}$. %
	Observe that from the proof of~\cref{infloopexist}, $\loopie$ is trivial if and only if $\Sat(L) = \ZZ^d$. This, in turn, is equivalent to $I_T = \ideal{\bm{x}^{\bm{\alpha}} - \bm{x}^{\bm{\beta}}  : \bm{\alpha}, \bm{\beta} \in \NN^d} = \ideal{x_1 - 1, \dots, x_d-1}$.
\end{proof}
\hl{Given polynomials $p_1, \dots, p_k$, a necessary condition for a non-trivial loop 
to simultaneously satisfy $p_i = 0$ for all $i =1, \dots, k$, 
is the existence of infinitely many rational solutions to a system of equations 
\(p_1= 0 \wedge \dots \wedge p_k= 0.\)
When~$p_1$, \dots, $p_k$ are pure difference binomials with~$k \leq d-1$, 
the system indeed has infinitely many solutions--and this is utilised in~\cref{infloopexist}. 
However, determining whether an arbitrary system of polynomial equations 
has infinitely many rational solutions is a highly non-trivial problem, 
as witnessed by the example below.
Note that determining whether an equation has infinitely many integer solutions is undecidable~\cite{davis1972numsol}.}
\begin{example}\label{cubessum1}Here, $d=2$ and $k=1$.
	a) The equation $x^3+y^3=1$ has precisely two rational solutions: $(x_1,y_1) = (1,0)$ and $(x_2,y_2) = (0,1)$. 
	b) For comparison, the equation $x^3+y^3 = 9$ has infinitely many rational solutions~\cite[Chapter 5.2]{silverman2015elliptic}.
\end{example}

\subsection{Modifications to the Synthesis Process}
	A higher objective than that set in~\cref{sec:binomial} is as follows:
	given a polynomial ideal \(I\), construct a linear loop whose invariant ideal is precisely \(I\).
	\hl{The invariant ideals are radical, and so, for this question, one may assume that the given ideal~$I$ is radical.	
	Equivalently, the objective is to construct a linear loop 
	whose Zariski closure is the variety~$V(I)$.}
	Our synthesis procedure does not generally achieve this goal because each inclusion in the chain of ideals
	\begin{equation}\label{chain} I \subseteq J \subseteq I_{L} \subseteq I_{\Sat(L)}\end{equation} 
	considered in~\cref{sec:binomial} is, in general, strict. 
	Here $J$ stands for the ideal generated by the canonical binomials, as in e.g.~\cref{IandJsat}.

\begin{remark} \label{rem:lattice}
	On the one hand, $\ker{A}$ of any integer $s\times d$-matrix~$A$ is a sublattice of~$\ZZ^d$.
	On the other hand, the converse statement is not true since there are lattices that cannot be represented as matrix kernels.
	For example, $L:= \ZZ(2, -2) \subset \ZZ^2$ is not saturated and hence not a kernel of a matrix: $(2, -2) \in L$ but $(1, -1) \not\in L$.
\end{remark}

By \autoref{rem:lattice}, we generally have $I_L \subsetneq I_{\Sat(L)}$. 
An examination of \cref{runex2} (see also~\cref{tblofex}), 
shows how most of the inclusions in~\eqref{chain} are strict; 
in this example, only the equality $J = I_L$ holds.
Nevertheless, its (general) violation is witnessed elsewhere in the literature (cf.~the discussion that follows \cite[Corollary 3.21]{herzog2018gtm}).

In special cases, we can simplify the procedure of~\cref{sec:binomial} by observing (introduced) redundancy. 
In this direction, we provide a %
list of sufficient conditions for equalities to hold in~\eqref{chain}.%
\begin{proposition}\label{shortcuts}Let $I \subseteq \KK[x_1, \dots, x_d]$ be a pure difference ideal.
	\begin{enumerate} 
		\item Suppose that $I = \ideal{p}$ is a principal ideal for which $p$ is a canonical pure difference binomial. 
		Then $I = I_L$. 
		If, in addition, $p$ is irreducible, then~$I = I_{\Sat(L)}$ holds.
		\item Suppose that $I$ is generated by pure difference binomials and at least one of the generators $p_1, \dots, p_k$ has a positive exponent vector~$\bm{v}_i \in \ZZ_{>0}^d$. Then, $J = I_L$.
	\end{enumerate}
\end{proposition}
\begin{proof}
	Regarding (1), we consider the principal ideal of~$p = \bm{x}^{\bm{\alpha}} - \bm{x}^{\bm{\beta}}$ with $\gcd(\bm{x}^{\bm{\alpha}},\bm{x}^{\bm{\beta}}) = 1$.
	Let~$I_L$ be the lattice ideal of~$L:= \ZZ \cdot (\bm{\alpha} - \bm{\beta})$. %
	This lattice ideal is generated by pure difference binomials.  
	Each such generator takes the form $\bm{x}^{n\bm{\alpha}} - \bm{x}^{n\bm{\beta}}$ for some positive integer~$n$, or $\bm{x}^{-m\bm{\beta}} - \bm{x}^{-m\bm{\alpha}}$ for some negative integer~$m$. 
	Polynomials of this form are divisible by $\bm{x}^{\bm{\alpha}} - \bm{x}^{\bm{\beta}}$. 
	It quickly follows that $I_L \subseteq \ideal{p}$ and hence~$I_L = I$. 
	If, in addition, $p$ is irreducible, 
	then $p$ is prime (\(\KK[x_1, \dots, x_d]\) is a unique factorisation domain). 
	Thus $I_L = I =\ideal{p}$ is prime and so~$L$ is saturated. Therefore $I = I_L = I_{\Sat(L)}$.
	
	Assertion (2), in turn, is a modification of~\cite[Lemma 12.4]{sturmfels1996gbcp}. Without loss of generality, let $J$ be an ideal generated by canonical pure difference binomials~$q_1, \dots, q_k$, where $q_1$'s exponent vector~$\bm{v}$ is positive. Then, $q_1 = \bm{x}^{\bm{v}}-1$ and every variable $x_i$, $1 \leq i \leq d$, has an inverse in the factor ring~$\KK[\bm{x}]/J$.
	Hence there exists $g \in \KK[\bm{x}]$ such that
	\( \textstyle
		p:= g \cdot \left(\prod_{i=1}^{d}x_i\right) - 1 \in J
	\).
	Let~$f \in \Col{J}$; that is, assume there exists~$n \geq 0$ such that $\left(\prod_{i=1}^{d}x_i \right)^n\cdot f \in J$.
	As a consequence, $g^n\cdot \left(\prod_{i=1}^{d}x_i \right)^n\cdot f \in J$.
	We rewrite the latter in the following form:
	\(\textstyle \left(g^n\left(\prod_{i=1}^{d}x_i \right)^n - 1 \right)\cdot f + f \in J
	\).
	Note that the expression wrapped by the outer parentheses is divisible by $p$. 
	It follows immediately that $f \in J$. Thus $J = \Col{J}$ and, by way of~\cref{ItoIL}, $I_L = J$.
\end{proof}
Note that~\cref{cor:canonical} follows from~\cref{shortcuts}~(i): 
an irreducible pure difference binomial is clearly canonical. 
A loop synthesised for~$I_{\Sat(L)}$ then has~$I$ as its invariant ideal by~\cref{toric-synt}.

\section{Discussion and Conclusions} \label{sec:conclusion}

\noindent\textbf{Related Works.} 
On the one hand,
deriving polynomial invariants for the class of loops with (non-linear) polynomial arithmetic in their update assignments is, in general, undecidable~\cite{Ouaknine20}.
On the other hand, for restricted classes of loops there are efficient procedures for synthesising invariants.
Indeed, the restricted classes of polynomial arithmetic in so-called \emph{solvable loops}~\cite{Rodriguez04}---loops with (blocks of) affine assignments, admit such procedures~\cite{Kovacs08,Oliveira16,Humenberger18,Kincaid18}.

Previous works synthesising linear loops from non-linear polynomial invariants include~\cite{humenberger2022journal}. 
Their approach is also based on algebraic reasoning about the C-finite
sequences generated by linear loops. Yet, the synthesis problem is
translated into a constraint solving task in non-linear arithmetic, by
relying on loop templates and using a conflict resolution procedure based on cylindrical
algebraic decomposition/Gr\"obner basis
computation~\cite{DBLP:journals/cca/JovanovicM12}. 
Unlike~\cite{humenberger2022journal}, 
our synthesis procedure restricts to linear algebraic reasoning. %
The algorithm of~\cite{humenberger2022journal} is relative complete as
it generates all linear loops satisfying a given invariant of a
bounded degree. 
As such, and in contrast to our results from \cref{infloopexist,inftoricideal}, 
the work of~\cite{humenberger2022journal} gives no sufficient conditions for the existence of at least one non-trivial loop satisfying the invariant. 

The approach in~\cite{galuppi2021toric} studies cyclic semigroups and 
sheds light on their geometry. 
In particular, irreducible components of 
Zariski closures for cyclic semigroups 
are shown to be isomorphic to toric varieties. 
Using the terminology in this note, their observation that
toric ideals are precisely the invariant ideals of some linear loops
is the key motivation for our synthesis procedure. 

\noindent\textbf{Conclusions and Future Directions.}
Using machinery from algebraic geometry, we identify classes of 
polynomial invariants 
for which synthesis becomes decidable.
Our synthesis procedure can
produce infinitely many different loops that, by construction, all
have the same invariant. 
For a given input, the number of variables $s$ of the procedurally generated
 output loops is equal to the number~$d$ of indeterminates in the input polynomials. 
Observe that linear loops with~$s > d$ variables can still satisfy invariants 
from~$\KK[x_1, \dots, x_d]$ on a subset of their variables. 
The challenge of deriving an upper bound on the number of program variables~$s$ for the invariant, also raised in~\cite[Section 4.2]{humenberger2022journal}, remains open. 

\hl{We aim to analyse the bit complexity of our procedure in the future. Note that in Step 3, prime numbers are raised to the powers that are, in general, exponential in the number of variables~$d$. Therefore, the generated update matrix might have entries with exponentially many (in~$D$) bits, cf.\ the worst-case double exponential running time of Gr\"obner basis	computation for the invariant ideal.}

We intent to extend our synthesis approach to synthesise loops 
whose invariant ideals are non-prime lattice ideals. 
One challenge to be addressed is the classification of Zariski closures of the orbits~$\OO$ of linear loops, cf.~\cite{galuppi2021toric}.
In this direction, it is of interest whether loops 
with non-diagonalisable matrices 
have invariant ideals
structurally different from those of loops synthesised in~\cref{sec:binomial,sec:to-bin}.

We conclude this paper by highlighting the following limitation of our
approach, which we intend to address in future work.
The construction in Step~1, that introduces the lattice ideal $I_L$,
entails that we can only synthesise loops whose invariant ideals are lattice ideals.
Further, \cref{ItoIL} implies that $V(I_L)$ is the Zariski closure of $V(I) \setminus V(\prod_{i=1}^{d}x_i)$, cf.~\cite[Chapter 4.4, Theorem 10]{cox2015ideals}. 
Our procedure thus only synthesises loops whose orbits lie outside of the coordinate hyperplanes. %
	\begin{example}\label{fourlines}
		We revisit the synthesis problem for an ideal~$I = \ideal{x^3y-xy^3}$ of~\cref{runex2}. 
		The invariant ideal of a linear loop~$\loopie'$ with update matrix%
		\[M:=\begin{pmatrix*}
			1 & -1\\
			1 & 1
		\end{pmatrix*}\]
		and initial vector~$\bm{v}_\text{init} = (1, 0)$ is precisely the ideal~$I$. 
	\end{example}
	The loop~$\loopie'$ in the example above stands in contrast to the synthesised loop~$\loopie$ in~\cref{runex2} 
	because~$V(I)$ is the strongest algebraic invariant of~$\loopie'$.
	In contrast, the strongest algebraic invariant of~$\loopie$ is $V(x-y)$.

\printbibliography

\end{document}